\documentclass[envcountsame,runningheads]{llncs}
\usepackage{pscproc2}
\usepackage{graphics}
\usepackage{algorithm}
\usepackage{amssymb}
\usepackage{textcomp}
\usepackage{amsmath}
\usepackage{algpseudocode}

\newcommand{\Endproof}{\hfill$\Box$\\}
\begin{document}
\title{Classical and Quantum Algorithms for Constructing Text from Dictionary Problem}

\author{Kamil~Khadiev$^{1,2,}$\thanks{This work was supported by Russian Science Foundation Grant 19-71-00149.} \and Vladislav Remidovskii$^{2}$}
\institute{$^1$ Smart Quantum Technologies Ltd., Kazan, Russia\\
$^2$ Kazan Federal University, Kazan, Russia\\
\email{kamilhadi@gmail.com,vladremidovskiyvs@gmail.com}
}
\maketitle

\begin{abstract}
We study algorithms for solving the problem of constructing a text (long string) from a dictionary (sequence of small strings). The problem has an application in bioinformatics and has a connection with the Sequence assembly method for reconstructing a long DNA sequence from small fragments.
The problem is constructing a string $t$ of length $n$ from strings $s^1,\dots, s^m$ with possible intersections. We provide a classical algorithm with running time $O\left(n+L +m(\log n)^2\right)=\tilde{O}(n+L)$ where $L$ is the sum of lengths of $s^1,\dots,s^m$. We provide a quantum algorithm with running time $O\left(n +\log n\cdot(\log m+\log\log n)\cdot \sqrt{m\cdot L}\right)=\tilde{O}\left(n +\sqrt{m\cdot L}\right)$. Additionally, we show that the lower bound for the classical algorithm is $\Omega(n+L)$. Thus, our classical algorithm is optimal up to a log factor, and our quantum algorithm shows speed-up comparing to any classical algorithm in a case of non-constant length of strings in the dictionary.  
\end{abstract}


\section{Introduction}
\emph{Quantum computing} \cite{nc2010,a2017,aazksw2019part1} is one of the hot topics in computer science of last decades.
There are many problems where quantum algorithms outperform the best known classical algorithms \cite{dw2001,quantumzoo,ks2019,kks2019}.

One of such problems are problems for strings. Researchers show the power of quantum algorithms for such problems in  \cite{m2017,bbbv1997,rv2003,ki2019}.

In this paper, we consider the problem of constructing text from dictionary strings with possible intersections. We have a text $t$ of length $n$ and a dictionary $s=s^1,\dots, s^m$. The problem is constricting $t$ only from strings of $s$ with possible intersections. The problem is connected with the sequence assembly method for reconstructing a long DNA sequence from small fragments \cite{msdd2000,bt2013}.

We suggest a classical algorithm  with running time $O\left(n+L +m(\log n)^2\right)=\tilde{O}(n+L)$, where $L=|s^1|+\dots+|s^m|$, $|s^i|$ is a length of $s^i$ and $\tilde{O}$ does not consider log factors. The algorithm uses segment tree\cite{l2017guide} and suffix array\cite{mm90} data structures, concepts of comparing string using rolling hash\cite{kr87,Fre79} and idea of prefix sum \cite{cormen2001}. 

The second algorithm is quantum. It  uses similar ideas and quantum algorithm for comparing two strings with quadratic speed-up comparing to classical counterparts \cite{ki2019}. The running time for our quantum algorithm is  

$O\left(n +\log n\cdot(\log m+\log\log n)\cdot \sqrt{m\cdot L}\right)=\tilde{O}\left(n +\sqrt{m\cdot L}\right)$.

Additionally, we show the lower bound in a classical case that is $\Omega(n+L)$. Thus, we get the optimal classical algorithm in a case of $m(\log n)^2=O(L+n)$.  It is true, for example, in a case of $O(m)$ strings form $s$ has length at least $\Omega\left((\log n)^2\right)$ or in a case of $m=o\left(n/(\log n)^2\right)$. In the general case, the algorithm is an optimal algorithm up to a log factor. 
The quantum algorithm is better than any classical counterparts in a case of  $ \log n\cdot(\log m+\log\log n)\cdot \sqrt{m\cdot L}=o(L)$. It happens if $O(m)$ strings from $s$ has length at least $\Omega(\log n\cdot(\log m+\log\log n))$.

 Our algorithm uses some quantum algorithms as a subroutine, and the rest part is classical. We investigate the problems in terms of query complexity. The query model is one of the most popular in the case of quantum algorithms. Such algorithms can do a query to a black box that has access to the sequence of strings. As a running time of an algorithm, we mean a number of queries to the black box.

The structure of the paper is the following. We present tools in Section \ref{sec:tools}. Then, we discuss the classical algorithm in Section \ref{sec:classical} and quantum algorithm in Section \ref{sec:quantum}. Section \ref{sec:lower} contains lower bound. 

\section{Tools}\label{sec:tools}

Our algorithms uses several data structures and algorithmic ideas like segment tree\cite{l2017guide}, suffix array\cite{mm90}, rolling hash\cite{kr87} and prefix sum \cite{cormen2001}. Let us describe them in this section. 
\subsection{Preliminaries}
Let us consider a string $u=(u_1,\dots,u_l)$ for some integer $l$. Then, $|u|=l$ is the length of the string. $u[i,j]=(u_i,\dots,u_j)$ is a substring of $u$. 

In the paper, we compare strings in lexicographical order. For two strings $u$ and $v$, the notation $u<v$ means $u$ precedes $v$ in lexicographical order.  
\subsection{Rolling Hash for Strings Comparing}\label{sec:roll-hash}
The rolling hash was presented in \cite{kr87}.
It is a hash function \[h_p(u)=\left(\sum\limits_{i=1}^{|u|}index(u_i)\cdot K^{i-1}\right)\mbox{ mod }p,\]
where $p$ is some prime integer, $K=|\Sigma|$ is a size of the alphabet and $index(u_i)\in\{0,\dots,K-1\}$ is the index of a symbol $u_i$ in the alphabet. For simplicity we consider binary alphabet. So, $K=2$ and $index(u_i)=u_i$.

We can use rolling hash and the fingerprinting method \cite{Fre79} for comparing two strings $u,v$. Let us  randomly choose $p$ from the set of the first $r$ primes, such that $r\leq \frac{\max(|u|,|v|)}{\varepsilon}$ for some $\varepsilon>0$. Due to Chinese Theorem and \cite{Fre79}, the following statements are equivalent $h_p(u)=h_p(v)$ and $u=v$ with error probability at most $\varepsilon$. If we have $\delta$ invocations of comparing procedure, then we should choose  $\frac{\delta \cdot \max(|u|,|v|)}{\varepsilon}$ primes. Due to Chebishev's theorem, the $r$-th prime number $p_r\approx r\ln r$. So, if our data type for integers is enough for storing $\frac{\delta \cdot \max(|u|,|v|)}{\varepsilon}\cdot (\ln(\delta) + \ln(\max(|u|,|v|))-\ln(\varepsilon))$, then it is enough for computing the rolling hash.


Additionally, for a string $u$, we can compute prefix rolling hash, that is
$h_p(u[1,i])$. It can be computed in $O(|u|)$ running time using formula
\[h_p(u[1,i])=\left(h_p(u[1,i-1])+(2^{i-1}\mbox{ mod }p)\cdot u_i\right)\mbox{ mod }p\mbox{ and }h_p(u[1:0])=0.\]
Assume, that we store ${\cal K}_i=2^{i-1}$ mod $p$. We can compute all of them in $O(|u|)$ running time using formula ${\cal K}_i={\cal K}_{i-1}\cdot 2$ mod $p$.

Using precomputed prefix rolling hash for a string $u$ we can compute rolling hash for any substring $u[i,j]$ in $O(1)$ running time by formula
\[h_p(u[i,j])=\left(\sum\limits_{q=j}^{i}u_q\cdot 2^{q-1-(j-1)}\right)\mbox{ mod }p=\left(\sum\limits_{q=j}^{i}u_q\cdot 2^{q-1}\right)\cdot 2^{-(j-1)}\mbox{ mod }p=\]
\[=\left(\left(\sum\limits_{q=1}^{i}u_q\cdot 2^{q-1}\right) - \left(\sum\limits_{q=1}^{j-1}u_q\cdot 2^{q-1}\right)\right)\cdot 2^{-(j-1)}\mbox{ mod }p =\]\[= \left(h_p(u[1,i])-h_p(u[1,j-1])\right)\cdot(2^{-(j-1)} ) \mbox{ $($mod }p).\]

For computing the formula in $O(1)$ we should precompute ${\cal I}_i=2^{-i}$ mod $p$. We can compute it in $O(\log p+|u|)$ by the formula ${\cal I}_i={\cal I}_{i-1}\cdot 2^{-1}$ mod $p$ and ${\cal I}_{0}=1$. Due to Fermat's little theorem $2^{-1}$ mod $p=2^{p-2}$ mod $p$. We can compute it with $O(\log p)$ running time using Exponentiation by squaring algorithm.

Let $\textsc{ComputeKI}(\beta,p)$ be a procedure that computes ${\cal K}$ and ${\cal I}$ up to the power $\beta$ with $O(\beta+\log p)$ running time. Let $\textsc{ComputePrefixRollingHashes}(u,p)$ be a procedure that computes all prefix rolling hashes for a string $u$ and store them.

Assume, that we have two strings $u$ and $v$ and already computed prefix rolling hashes. Then, we can compare these strings in lexicographical order in $O(\log \min(|u|,|v|))$ running time. The algorithm is following. We search the longest common prefix of $u$ and $v$, that is $lcp(u,v)$. We can do it using binary search.
\begin{itemize}
    \item If a $mid\leq lcp(u,v)$, then $h_p(u[1,mid])=h_p(v[1,mid])$.
    \item If a $mid> lcp(u,v)$, then $h_p(u[1,mid])\neq h_p(v[1,mid])$.
\end{itemize}

Using binary search we find the last index $x$ such that $h_p(u[1,x])=h_p(v[1,x])$ and $h_p(u[1,x+1])\neq h_p(v[1,x+1])$. In that case $lcp(u,v)=x$

  After that, we compare $u_{t}$ and $v_t$ for $t=lcp(u,v)+1$. If $u_t<v_t$, then $u<v$; if $u_t>v_t$, then $u>v$; if $|u|=|v|=lcp(u,v)$, then $u=v$.

Binary search works with $O(\log (\min(|u|,|v|)))$ running time because we have computed all prefix rolling hashes already.

Let $\textsc{Compare}(u,v)$ be a procedure that compares $u$ and $v$ and returns $-1$ if $u<v$; $1$ if $u>v$; and $0$ if $u=v$.
\subsection{Segment Tree with Range Updates}\label{sec:segment-tree}

We consider a standard segment tree data structure \cite{l2017guide} for an array $a=(a_1,\dots, a_l)$ for some integer $l$. A segment tree for and array $a$ can be constructed in $O(l)$ running time. The data structure allows us to invoke the following requests in $O(\log l)$ running time.

\begin{itemize}
    \item {\bf Update.} Parameters are three integers $i, j, x$ ($1\leq i\leq j\leq l$). We assign $a_q=\max(a_q,x)$ for $i\leq q\leq j$. 
     \item {\bf Push.} We push all existing range updates. 
    \item {\bf Request.} For an integer $i$ ($1\leq i\leq l$), we should return $a_i$. 
\end{itemize}

Let $\textsc{ConstructSegmentTree}(a)$ be a function that constructs and returns a segment tree for an array $a$ in $O(l)$ running time. 

Let $\textsc{Update}(st,i,j,x)$ be a procedure that updates a segment tree $st$ in $O(\log l)$ running time.

Let $\textsc{Push}(st)$ be a procedure that push all existing range updates for a segment tree $st$ in $O(l)$ running time.

Let $\textsc{Request}(st,i)$ be a function that returns $a_i$ from a segment tree $st$. The running time of the procedure is $O(\log l)$. At the same time, if we invoke $\textsc{Push}$ procedure and after that do not invoke $\textsc{Update}$ procedure, then the running time of $\textsc{Request}$ is $O(1)$.
\subsection{Suffix Array}
Suffix array \cite{mm90} is an array $suf=(suf_1,\dots,suf_{l})$ for a string $u$ and $l=|u|$. The suffix array is a lexicographical order for all suffixes of $u$. Formally, $u[suf_i,l]<u[suf_{i+1},l]$ for any $i\in\{1,\dots,l-1\}$.

The suffix array can be computed in $O(l)$ running time.

\begin{lemma}[\cite{llh2018}]\label{lm:suf-arr}
A suffix array for a string $u$ can be constructed in $O(|u|)$ running time.
\end{lemma}  

Let $\textsc{ConstructSuffixArray}(u)$ be a procedure that constructs a suffix array for a string $u$.

\section{Algorithms}\label{sec:algos}
Let us formally present the problem.

{\bf Problem.}
For some positive integers $n$ and $m$, we have a sequence of strings $s=(s^1,\dots,s^m)$. Each $s^i=(s^i_1,\dots,s^i_{l})\in \Sigma^l$ where $\Sigma$ is some finite size alphabet and $l=|s^i|$.  We call $s$ dictionary. Additionally, we have a string $t$ of length $n$, where $t=(t_1,\dots,t_n)\in \Sigma^n$. We call $t$ text. The problem is searching a subsequence $s^{i_1},\dots,s^{i_z}$ and positions $q_1,\dots,q_z$ such that $q_1=1$, $q_z=n-|s^{i_z}|+1$, $q_j\leq q_{j-1}+|s^{i_{j-1}}|$  for $j\in\{2,\dots,z\}$. Additionally, $t[q_j,q_j+|s^{i_j}|-1]=s^{i_j}$ for $j\in\{1,\dots,z\}$.

For simplicity, we assume that $\Sigma=\{0,1\}$, but all results are right for any finite alphabet.
 
Informally, we want to construct $t$ from $s$ with possible intersections. 
 
Firstly, let us present a classical algorithm.
\subsection{A Classical Algorithm}\label{sec:classical}
Let us present the algorithm.
Let $long_i$ be an index of a longest string from $s$ that can start in position $i$. Formally, $long_i=j$ if $s^j$ is a longest string from $s$ such that $t[i,i+|s^j|-1]=s^j$. Let $long_i=-1$ if there is no such string $s^j$.
If we construct such array, then we can construct $Q=(q_1,\dots,q_z)$ and $I = (i_1,\dots,i_z)$ that is solution of the problem in $O(n)$. A procedure $\textsc{ConstructQI}(long)$ from Algorithm \ref{alg:qi} shows it. If there is no such decomposition of $t$, then the procedure returns $NULL$.

\begin{algorithm}
\caption{$\textsc{ConstructQI}(long)$. Constructing $Q$ and $I$ from $last$}\label{alg:qi}
\begin{algorithmic}
\State $t \gets 1$
\State $i_1\gets long_1$
\State $q_1\gets 1$
\State $left\gets 2$
\State $right\gets |s^{i_1}|+1$
\While{$q_t<n$}
\State{$max\_i\gets left$}
\State{$max\_q\gets -1$}
\If{$long_{left}>0$}
\State{$max\_q\gets left + |s^{long_{left}}|-1$}
\EndIf
\For{$j\in \{left+1,\dots,right\}$}
\If{$long_j>0$ and $j + |s^{long_j}|-1>max\_q$}
\State{$max\_i\gets j$}
\State{$max\_q\gets j + |s^{long_j}|-1$}
\EndIf
\EndFor
\If{$max\_q= -1$ or $max\_q< right$}
\State Break the While loop and \Return $NULL$ \Comment{We cannot construct other part of the string $t$}
\EndIf 
\State $t\gets t+1$
\State $i_t\gets long_{max\_i}$
\State $q_t\gets max\_i$
\State $left\gets right+1$
\State $right\gets max\_q+1$
\EndWhile
\State \Return $(Q,I)$
\end{algorithmic}
\end{algorithm}

Let us discuss how to construct $long$ array.

As a first step, we choose a prime $p$ that is used for rolling hash. We choose $p$ randomly from the first $z=n\cdot m \cdot 4\lceil\log_2 n\rceil^2\cdot \frac{1}{\varepsilon}$ primes. In that case, due to results from Section \ref{sec:roll-hash}, the probability of error  is at most $\varepsilon$ in a case of at most $m \cdot 4\lceil\log_2 n\rceil$ strings comparing invocations.
 
As a second step, we construct a suffix array $suf$ for $t$. Then, we consider an array $a$ of pairs $(len, ind)$. One element of $a$ corresponds to one element of the suffix array $suf$. After that, we construct a segment tree $st$ for $a$ and use $len$ parameter of pair for maximum.

As a next step, we consider strings $s^i$ for $i\in\{1,\dots m\}$. For each string $s^i$ we find the smallest index $low$ and the biggest index $high$ such that all suffixes $t[suf_j,n]$ for $low\leq j \leq high$ has $s^i$ as a prefix. We can use binary search for this action. Because of sorted order of suffixes in suffix array, all suffixes with the prefix $s^i$ are situated sequently. 
As a comparator for strings, we use $\textsc{Compare}$ procedure.
Let us present this action as a procedure $\textsc{SearchSegment}(u)$ in Algorithm \ref{alg:search}. The algorithm returns $(NULL,NULL)$ if no suffix of $t$ contains $u$ string as a prefix.

\begin{algorithm}
\caption{$\textsc{SearchSegment}(u)$. Searching a indexes segment  of suffixes for $t$ that have $u$ as a prefix}\label{alg:search}
\begin{algorithmic}
\State $low\gets NULL$, $high\gets NULL$
\State $l\gets 1$
\State $r\gets n$
\State $Found\gets False$
\While{$Found=False$ and $l\leq r$}
\State $mid\gets (l+r)/2$
\State $pref\gets t[suf_{mid},\min(n,suf_{mid}+|u|-1)]$
\State $pref1\gets t[suf_{mid-1},\min(n,suf_{mid-1}+|u|-1)]$
\State $compareRes \gets \textsc{Compare}(pref,u), compareRes1 \gets \textsc{Compare}(pref1,u)$
\If {$compareRes=0$ and $compareRes1=-1$}
\State $Found\gets true$
\State $low\gets mid$
\EndIf
\If {$compareRes< 0$}
\State $l\gets mid+1$
\EndIf
\If {$compareRes\geq 0$}
\State $r\gets mid-1$
\EndIf
\EndWhile
\If{$Found=True$}
\State $l\gets 1$
\State $r\gets n$
\State $Found\gets False$
\While{$Found=False$ and $l\leq r$}
\State $mid\gets (l+r)/2$
\State $pref\gets t[suf_{mid},\min(n,suf_{mid}+|u|-1)]$
\State $pref1\gets t[suf_{mid+1},\min(n,suf_{mid+1}+|u|-1)]$
\State $compareRes \gets \textsc{Compare}(pref,u), compareRes1 \gets \textsc{Compare}(pref1,u)$
\If {$compareRes=0$ and $compareRes1=+1$}
\State $Found\gets true$
\State $high\gets mid$
\EndIf
\If {$compareRes\leq 0$}
\State $l\gets mid+1$
\EndIf
\If {$compareRes> 0$}
\State $r\gets mid-1$
\EndIf
\EndWhile
\EndIf
\State \Return $(low, high)$
\end{algorithmic}
\end{algorithm}

Then, we update values in the segment tree $st$ by a pair $(|s^i|,i)$.

After processing all strings from $(s^1,\dots,s^m)$, the array $a$ is constructed. We can construct $long$ array using $a$ and the suffix array $suf$. We know that  $i$-th element $(len,ind)$ stores the longest possible string $s^{ind}$ that starts from $suf_i$. It is almost definition of $long$ array. So we can put $long_{suf_i}=ind$, if $a_i=(ind,len)$.

Finally, we get the following Algorithm \ref{alg:classical} for the text constructing from a dictionary  problem.
  
\begin{algorithm}
\caption{The classical algorithm for the text $t$ constructing from a dictionary $s$ problem for an error probability $\varepsilon>0$}\label{alg:classical}
\begin{algorithmic}
\State $\alpha\gets m \cdot 4\lceil\log_2 n\rceil^2$
\State $r\gets n\cdot \alpha/\varepsilon$
\State $p\in_R \{p_1,\dots,p_r\}$\Comment{$p$ is randomly chosen from $\{p_1,\dots,p_r\}$}

\State $\textsc{ComputeKI}(n,p)$
\State $\textsc{ComputePrefixRollingHashes}(t,p)$
\For{$j\in\{1,\dots,m\}$}
\State $\textsc{ComputePrefixRollingHashes}(s^j,p)$
\EndFor
\State $suf\gets\textsc{ConstructSuffixArray}(t)$
\State $a\gets[(0,-1),\dots,(0,-1)]$\Comment{Initialization by $0$-array}
\State $st \gets\textsc{ConstructSegmentTree}(a)$
\For{$j\in\{1,\dots,m\}$}
\State $(low,high)\gets \textsc{SearchSegment}(s^j)$
\If{$(low,heigh)\neq (NULL,NULL)$}
\State $\textsc{Update}(st,low,high,(|s^j|,j))$
\EndIf
\EndFor
\State $\textsc{Push}(st)$
\For{$i\in\{1,\dots,n\}$}
\State $(len,ind)\gets \textsc{Request}(st,i)$
\State $long_{suf_i}\gets ind$
\EndFor
\State $(Q,I)\gets\textsc{ConstructQI}(long)$
\State\Return $(Q,I)$
\end{algorithmic}
\end{algorithm}

Let us discuss properties of  Algorithm \ref{alg:classical}.
\begin{theorem}\label{th:classical}
Algorithm \ref{alg:classical} solves the text $t$ constructing from a dictionary $s=(s^1,\dots,s^m)$ problem with $O\left(n+L +m(\log n)^2 -\log\varepsilon\right)$ running time end error probability $\varepsilon$ for some $\varepsilon>0$, $n=|t|$ and $L=|s^1|+\dots + |s^m|$. The running time is $O\left(n+L +m(\log n)^2\right)$ in a case of $\varepsilon=const$.
\end{theorem}
\begin{proof}
The correctness of the algorithm follows from construction.

Let us discuss running time of the algorithm. Note, that $m\leq L=\sum_{j=1}^m|s^j|$ and $1\leq |s^j|\leq n$ for $j\in\{1,\dots,m\}$.

 Due to results from Section \ref{sec:roll-hash}, $\textsc{ComputeKI}$ works with $O(n+\log p)$ running time. Let us convert this statement.
\[O(n+\log p)=O(n+\log (n\cdot \alpha/\varepsilon))=\]
\[=O(n + \log n + \log\alpha - \log\varepsilon)=O(n + \log\alpha - \log\varepsilon)=\]\[=
O(n + \log(m \cdot 4\lceil\log_2 n\rceil^2) - \log\varepsilon)
=O(n + \log m  +\log\log n - \log\varepsilon)=\]\[=
O(n+\log m - \log\varepsilon)\]

Due to results from Section \ref{sec:roll-hash}, $\textsc{ComputePrefixRollingHashes}$ works in linear running time. Therefore,   all invocations of $\textsc{ComputePrefixRollingHashes}$ procedure works in $O(n+\sum_{j=1}^m|s^j|)=O(n+L)$ running time.

Due to Lemma \ref{lm:suf-arr}, $\textsc{ConstructSuffixArray}$ works in $O(n)$ running time. Initializing of $a$ does $O(n)$ steps. Due to results from Section \ref{sec:segment-tree}, $\textsc{ConstructSegmentTree}$ works in $O(n)$ running time.

$\textsc{SearchSegment}$ procedure invokes $\textsc{Compare}$ procedure $O(\log n)$ times due to binary search complexity. $\textsc{Compare}$ procedure works in $O(\log n)$ running time. Therefore,  $\textsc{SearchSegment}$ works in $O((\log n)^2)$ running time.
Due to results from Section \ref{sec:segment-tree}, $\textsc{Update}$ procedure works in $O(\log n)$ running time.
Hence, the total complexity of processing all strings from the dictionary $s$ is $O\left(m\cdot((\log n)^2+\log n)\right)=O\left(m\cdot(\log n)^2\right)$.

The invocation of $\textsc{Push}$ works in $O(n)$ running time due to results from Section \ref{sec:segment-tree}.
The invocation of $\textsc{Request}$ works in $O(1)$ running time because we do not invoke $\textsc{Update}$ after $\textsc{Push}$. Therefore, constructing of the array $long$ takes $O(n)$ steps.

The running time of $\textsc{ConstructQI}$ is $O(n)$ because we pass each element only once.

So, the total complexity of the algorithm is
\[O\left(n+\log m - \log\varepsilon+n+L+n+n+n+m(\log n)^2 +n+n+n\right)=\]\[=O\left(n+L +m(\log n)^2 -\log\varepsilon\right).\]

Let us discuss the error probability. We have $4 \cdot m\lceil\log_2 n\rceil$ invocations of $\textsc{Compare}$ procedure. Each invocation of $\textsc{Compare}$ procedure compares rolling hashes at most $\lceil\log_2 n\rceil$ times. Due to results from Section \ref{sec:roll-hash}, if we compare strings of length at most $n$ using rolling hash $4 \cdot m\lceil\log_2 n\rceil^2$ times and choose $p$ from $r$ primes, then we get error probability at most $\varepsilon$.
\Endproof
\end{proof}

\subsection{A Quantum Algorithm}\label{sec:quantum}
Firstly, let us discuss a quantum subroutine.
There is a quantum algorithm for comparing two strings in a lexicographical order with the following property:

\begin{lemma}[\cite{ki2019}]\label{lm:str-compr}There is a quantum algorithm that compares two strings of length $k$ in lexicographical order with  query complexity $O(\sqrt{k}\log \gamma)$ and error probability $O\left(\frac{1}{ \gamma^3}\right)$ for some positive integer $\gamma$.
\end{lemma}

Let $\textsc{QCompare\_strings\_base}(u,v,l)$ be a quantum subroutine for comparing two strings $u$ and $v$ of length $l$ in lexicographical order. We choose $\gamma = m\log n$. In fact, the procedure compares prefixes of $u$ and $v$ of length $l$.
 $\textsc{QCompare\_strings\_base}(u,v,l)$ returns $1$ if $u>v$; it returns $-1$ if $u<v$; and it returns $0$ if $u=v$.

Next, we use a $\textsc{QCompare}(u,v)$ that compares $u$ and $v$ string in lexicographical order. Assume that $|u|<|v|$. Then, if $u$ is a prefix of $v$, then $u<v$. If $u$ is not a prefix of $v$, then the result is the same as for $\textsc{QCompare\_strings\_base}(u,v,|u|)$. In the case of $|u|>|v|$, the algorithm is similar. The idea is presented in Algorithm \ref{alg:compare}.

\begin{algorithm}
\caption{The quantum algorithm for comparing two string in lexicographical order }\label{alg:compare}
\begin{algorithmic}
\If{$|u|=|v|$}
\State $Result \gets\textsc{QCompare\_strings\_base}(u,v,|u|) $
\EndIf
\If{$|u|<|v|$}
\State $Result \gets\textsc{QCompare\_strings\_base}(u,v,|u|) $
\If{$Result=0$}
\State $Result = -1$
\EndIf
\EndIf
\If{$|u|>|v|$}
\State $Result \gets\textsc{QCompare\_strings\_base}(u,v,|v|) $
\If{$Result=0$}
\State $Result = 1$
\EndIf
\EndIf
\State \Return $Result$
\end{algorithmic}
\end{algorithm}

Let us present a quantum algorithm for the text constructing form a dictionary problem. For the algorithm, we use the same idea as in the classical case, but we replace $\textsc{Compare}$ that uses the rolling hash function by $\textsc{QCompare}$. In that case, we should not construct rolling hashes. Let $\textsc{QSearchSegment}$ be a quantum counterpart of $\textsc{SearchSegment}$ that uses $\textsc{QCompare}$.

The quantum algorithm is presented as Algorithm \ref{alg:quantum}.

\begin{algorithm}
\caption{The quantum algorithm for the text $t$ constructing from a dictionary $s$ problem}\label{alg:quantum}
\begin{algorithmic}
\State $suf\gets\textsc{ConstructSuffixArray}(t)$
\State $a\gets[(0,-1),\dots,(0,-1)]$\Comment{Initialization by $0$-array}
\State $st \gets\textsc{ConstructSegmentTree}(a)$
\For{$j\in\{1,\dots,m\}$}
\State $(low,high)\gets \textsc{QSearchSegment}(s^j)$
\State $\textsc{Update}(st,low,high,(|s^j|,j))$
\EndFor
\State $\textsc{Push}(st)$
\For{$i\in\{1,\dots,n\}$}
\State $(len,ind)\gets \textsc{Request}(st,i)$
\State $long_{suf_i}\gets ind$
\EndFor
\State $(Q,I)\gets\textsc{ConstructQI}(long)$
\State\Return $(Q,I)$
\end{algorithmic}
\end{algorithm}

Let us discuss properties of Algorithm \ref{alg:quantum}.
\begin{theorem}
Algorithm \ref{alg:quantum} solves the text $t$ constructing from a dictionary $s=(s^1,\dots,s^m)$ problem in $O\left(n +\log n\cdot(\log m+\log\log n)\cdot \sqrt{m\cdot L}\right)$ running time end error probability $O\left(\frac{1}{m+\log n}\right)$. 
\end{theorem}
\begin{proof}
The algorithm does almost the same actions as the classical counterpart. That is why the correctness of the algorithm follows from Theorem \ref{th:classical}.

Let us discuss the running time. Due to Theorem \ref{th:classical}, the running time of the procedure $\textsc{ConstructSuffixArray}$ is $O(n)$, the running time of the procedure $\textsc{ConstructSegmentTree}$ is $O(n)$, the running time of the procedure $\textsc{Push}$ is $O(n)$, the running time of the array $long$  construction is $O(n)$, the running time of $\textsc{ConstructQI}$ is $O(n)$.

Due to Lemma \ref{lm:str-compr}, the running time of $\textsc{QCompare}$ for $s^j$ is $O(\sqrt{|s^j|}(\log m+\log\log n))$. The procedure $\textsc{QSearchSegment}$ invokes $\textsc{QCompare}$ procedure $O(\log n)$ times for each string $s^1,\dots s^m$. So, the complexity of processing all strings from $s$ is

\[O\left(\log n\cdot(\log m+\log\log n)\cdot \sum_{j=1}^m\sqrt{|s^j|}\right)\]

Let us use the Cauchy-Bunyakovsky-Schwarz inequality and $L=\sum_{j=1}^m|s^j|$ equality for simplifying the statement.

\[\leq O\left(\log n\cdot(\log m+\log\log n)\cdot \sqrt{m\sum_{j=1}^m|s^j|}\right)=O\left(\log n\cdot(\log m+\log\log n)\cdot \sqrt{m\cdot L}\right).\]

The total running time is \[O\left(n+n+n +\log n\cdot(\log m+\log\log n)\cdot \sqrt{m\cdot L}\right)=\]\[=O\left(n +\log n\cdot(\log m+\log\log n)\cdot \sqrt{m\cdot L}\right)\]

Let us discuss the error probability. The algorithm invokes $\textsc{QCompare}$ procedure $2m\lceil\log_2 n \rceil\leq 2m(1+\log_2 n)$ times.
The success probability is \[O\left(\left(1-\frac{1}{(m\log n)^3}\right)^{2m(1+\log_2 n)}\right)=O\left(\frac{1}{m\log n}\right).\]
\Endproof
\end{proof}

\section{Lower Bound}\label{sec:lower}

Let us discuss the lower bound for the running time of classical algorithms.

\begin{theorem}\label{th:dfreq-compl}
Any randomized algorithm for the text $t$ constructing from a dictionary $s=(s^1,\dots,s^m)$ problem works in $\Omega(n+L)$ running time, where $L=|s^1|+\dots+|s^m|$.
\end{theorem}
\begin{proof}
Assume $L>n$.
Let us consider $t=(t_1,\dots,t_n)$ such that $t_{\lfloor n/2\rfloor}=1$ and $t_{i}=0$ for all $i\in\{1,\dots, n\}\backslash\{\lfloor n/2\rfloor\}$.

Let $|s^i|\leq n/2$ for each $i\in\{1,\dots,m\}$. Note, that in a general case $|s^i|\leq n$. Therefore, we reduce the input data only at most twice.
Assume that we have two options: 
\begin{itemize}
\item all $s^i$ contains only $0$s;
\item there is $z$ such that we have two conditions:
\begin{itemize}
\item all $s^i$ contains only $0$s, for $i\in\{1,\dots, m\}\backslash\{ z\}$;
\item for $j_0<|s^z|$, $s^z_{j_0}=1$ and $s^z_{j}=0$ for $j\in\{1,\dots, |s^z|\}\backslash\{ j_0\}$.
\end{itemize}
\end{itemize}
In a case of all $0$s, we cannot construct the text $t$. In a case of existing $0$ in the first half, we can construct $t$ by putting $s^z$ on the position $j_0-\lfloor n/2\rfloor+1$ and we get $1$ of the required position. Then, we complete $t$ by other $0$-strings.

Therefore, the solution of the problem is equivalent to the search of $1$ in unstructured data of size $L$. The randomized complexity of this problem is $\Omega(L)$ due to \cite{bbbv1997}. 

Assume $L<n$. Let $m=1$, $|s^1|=1$ and $s^1_1=1$.
Assume that we have two options:
\begin{itemize}
    \item $t$ contains only $1$s;
    \item there is $g$ such that $t_g=0$ and $t_j=1$ for all $j\in\{1,\dots,n\}\backslash\{g\}$.
\end{itemize}

In the first case, we can construct $t$ from $s$. In the second case, we cannot do it.

Here the problem is equivalent to search $0$ among $n$ symbols of $t$. Therefore, the problem's randomized complexity is $\Omega(n)$.

Finally, the total complexity is $\Omega(\max(n,L))=\Omega(n+L)$.
\Endproof
\end{proof}

\bibliographystyle{psc}          
\bibliography{tcs}          

\end{document}